\documentclass[11pt]{article}

\usepackage{xcolor} 
\usepackage{authblk}
\usepackage{amsfonts, amsmath, amssymb, amsthm}
\usepackage{array}
\newcolumntype{L}[1]{>{\raggedright\let\newline\\\arraybackslash\hspace{0pt}}m{#1}}
\newcolumntype{C}[1]{>{\centering\let\newline\\\arraybackslash\hspace{0pt}}m{#1}}
\newcolumntype{R}[1]{>{\raggedleft\let\newline\\\arraybackslash\hspace{0pt}}m{#1}}
\usepackage[english]{babel}
\usepackage{cmap} 
\usepackage{enumitem}
\setlist[itemize]{noitemsep}

\usepackage[T1]{fontenc}
\usepackage{graphicx}
\usepackage[a4paper]{geometry}
\usepackage[breaklinks,colorlinks = true, urlcolor = blue]{hyperref}
\usepackage[utf8]{inputenc}
\usepackage{nccmath}
\usepackage{stmaryrd} 
\SetSymbolFont{stmry}{bold}{U}{stmry}{m}{n} 
\allowdisplaybreaks

\newenvironment{rcases}
  {\left.\begin{aligned}}
  {\end{aligned}\right\rbrace}

\renewcommand{\phi}{\varphi}
\newcommand{\NN}{\mathbb{N}}
\newcommand{\intinter}[2]{\llbracket #1, #2 \rrbracket}
\makeatletter
\DeclareRobustCommand\Supp{\mathop{\operator@font Supp}\nolimits}
\makeatother
\newcommand{\antiport}[2]{( #1, \text{out}\ | \ #2, \text{in} )}
\newcommand{\antiportbr}[2]{( #1, \text{out}\ | \\ & \qquad \qquad \ #2, \text{in} )}
\newcommand{\sympout}[1]{( #1, \text{out} )}
\newcommand{\sympin}[1]{( #1, \text{in} )}
\newcommand{\CSC}{\mathbf{CSC}}
\newcommand{\CSCk}[1]{\mathbf{CSC}(#1)}
\newcommand{\TSC}{\mathbf{TSC}}

\newcommand{\PMC}{\mathbf{PMC}}
\newcommand{\PMCCSC}{\PMC_{\CSC}}
\newcommand{\PMCCSCk}[1]{\PMC_{\CSCk{#1}}}
\newcommand{\PMCTSC}{\PMC_{\TSC}}

\renewcommand{\P}{\mathbf{P}}
\newcommand{\NP}{\mathbf{NP}}
\newcommand{\PsharpP}{\mathbf{P}^\mathbf{\#P}}
\newcommand{\PSPACE}{\mathbf{PSPACE}}
\newcommand{\SAT}{\mathbf{SAT}}
\newcommand{\MIDSAT}{\mathbf{MIDSAT}}
\newcommand{\DeltaSAT}{\Delta_{SAT}}
\newcommand{\DeltaACT}{\Delta_{ACT}}
\newcommand{\DeltaSEA}{\Delta_{SEA}}
\newcommand{\DeltaOUT}{\Delta_{OUT}}
\newcommand{\DeltaPrep}{\Delta_{Prep}}
\newcommand{\DeltaComp}{\Delta_{Comp}}
\newcommand{\DeltaDete}{\Delta_{Dete}}
\newcommand{\DeltaLoop}{\Delta_{Loop}}
\newcommand{\IPrep}{I_{Prep}}
\newcommand{\IComp}{I_{Comp}}
\newcommand{\IDete}{I_{Dete}}
\newcommand{\IOUT}{I_{OUT}}

\theoremstyle{definition}
\newtheorem{definition}{Definition}

\theoremstyle{plain}
\newtheorem{theorem}{Theorem}
\newtheorem{proposition}{Proposition}
\newtheorem{lemma}{Lemma}
\newtheorem{corollary}{Corollary}

\theoremstyle{remark}
\newtheorem{remark}{Remark}
\newtheorem{example}{Example}

\begin{document}

\title{\texorpdfstring{Symport/Antiport P Systems with\\ Membrane Separation Characterize $\PsharpP$}{Symport/Antiport P Systems with Membrane Separation Characterize P\textasciicircum(\#P)}}

\author[1]{Vivien Ducros}
\affil[1]{ENS Paris-Saclay, Université Paris-Saclay, Gif-sur-Yvette, France}

\author[2]{Claudio Zandron}
\affil[2]{Dipartimento di Informatica, Sistemistica e Comunicazione, Universit\`a degli Studi di Milano-Bicocca, Milano, Italy}

\maketitle

\begin{abstract}
Membrane systems represent a computational model that operates in a distributed and parallel manner, inspired by the behavior of biological cells.
These systems feature objects that transform within a nested membrane structure.
This research concentrates on a specific type of these systems, based on cellular symport/antiport communication of chemicals.

Results in the literature show that systems of this type that also allow cell division can solve $\PSPACE$ problems.
In our study, we investigate systems that use membrane separation instead of cell division, for which only limited results are available.
Notably, it has been shown that any problem solvable by such systems in polynomial time falls within the complexity class $\PsharpP$.

By implementing a system solving $\MIDSAT$, a $\PsharpP$-complete problem, we demonstrate that the reverse inclusion is true as well, thus providing an exact characterization of the problem class solvable by P systems with symport/antiport and membrane separation.

Moreover, our implementation uses rules of length at most three.
With this limit, systems were known to be able to solve $\NP$-complete problems, whereas limiting the rules by length two, they characterize $\P$.
\end{abstract}

\section{Introduction}

Membrane systems (or P~systems) were introduced in \cite{Paun00}.
They are a computational model inspired by the operations made by biological cells on different types of chemicals.
The structure of these systems comprises a hierarchy of nested membranes, each defining distinct regions.
A system is enclosed in an outermost membrane called the skin, isolating it from the external environment.
Within these regions, objects undergo transformations based on specific evolution rules, enabling their movement and communication across membranes between different regions.

An interesting and deeply investigated feature of such systems was proposed shortly after in \cite{PaunNP01}, where the concept of \emph{active membranes} was introduced: membranes can be \emph{divided} by evolutionary rules, duplicating their contents in both obtained copies, thus allowing for the generation of an exponential amount of resources within a polynomial time frame.
An alternative approach, later considered, is \emph{separation} of membranes \cite{Alhazov2004}, where membranes are still duplicated but the contents of the original membrane are distributed between the two resulting membranes rather than being duplicated.

Since then, many studies have investigated the complexity classes defined by different types of P systems with active membranes, employing various features and constraints.
These studies aim to understand how specific properties influence the development of time-efficient systems that can solve computationally difficult problems in polynomial time while using exponential space, in contrast to less efficient systems.
In this paper, we continue this line of research by concentrating on a variant of membrane systems known as symport/antiport P systems, introduced in \cite{PP02}.
In this model, objects can be communicated between two regions in one of two ways: either by moving multiple objects simultaneously from the same region to an adjacent one (symport) or by exchanging two or more objects between two adjacent regions (antiport).

In \cite{Song15}, authors proved that symport/antiport P systems with membrane division and rules of length (the number of objects involved in the rule) not exceeding 3 can uniformly solve the well-known $\PSPACE$-complete problem $\mathbf{QSAT}$, in polynomial time and exponential space.
In \cite{Sosik19}, the authors conjectured that the reverse inclusion also holds, which would provide a characterization of $\PSPACE$ in terms of these systems.

Regarding symport/antiport P systems with membrane separation, only partial results have been obtained so far.
For example, they are known to solve the $\SAT$ problem \cite{VCAl15}, indicating that $\mathbf{NP}$ is a lower bound for these systems.
The influence of the environment has been studied in \cite{OMAl20}, showing that while it does not affect the computational power of symport/antiport P systems with membrane division, removing the environment in the case of membrane separation reduces the class of solvable problems to $\mathbf{P}$.
When the length of the rules is restricted to 1 (resp. 2), symport/antiport P systems with membrane division (resp. separation) are inefficient and can only solve problems in $\mathbf{P}$ \cite{MRAl17} (resp. \cite{MRAl15}).

In all these instances, the membranes are arranged in a nested structure.
However, another variant, known as \emph{tissue P systems}, arranges the membranes as a directed graph.
When fission (division or separation) rules are allowed, these systems characterize the class $\PsharpP$ of problems solved by polynomial time Turing Machine using $\mathbf{\#P}$ oracles \cite{LepoAl17}.
By exploiting this last result and the relations between cell-like and tissue-like P systems making use of membrane separation, we first show in this work that $\PsharpP$ represents an upper bound for P systems with symport/antiport and membrane separation.

Moreover, we prove that the opposite inclusion also holds by solving the $\PsharpP$-complete problem $\MIDSAT$, thus obtaining a precise characterization of the class of problems solved by such kind of P systems.
The $\MIDSAT$ problem takes a boolean formula $\phi(x_1,\dots,x_n)$ as input, and decides whether the variable $x_n$ is assigned to $1$ in the lexicographically middle satisfying assignment of $\phi$.
The system solving $\MIDSAT$ we designed is based on the one solving $\SAT$ from \cite{VCAl15}.
We provide an implementation of the system in the \texttt{MeCoSim} application, a simulator of P systems \cite{PLingua}, with the specific model for symport/antiport systems \cite{PLinguaSA}.

The rest of the paper is organized as follows: section \ref{section:definitions} is dedicated to the definitions of symport/antiport P systems and their complexity classes.
In section~\ref{section:results}, we give some results towards the characterization (Corollary \ref{corollary:characterization}).
The system solving $\MIDSAT$ is described in section \ref{section:system} and we prove it is correct.
Finally, we add a few comments about our implementation in \texttt{MeCoSim} in section~\ref{section:implementation}.

\section{Definitions}
\label{section:definitions}

We give in this section the usual definitions of P systems and complexity classes of problems solved by P systems.
For more details, our reader can refer to the introduction and to the handbook of membrane computing \cite{Paun06}, \cite{HBMC10}.
The notions of complexity theory for membrane systems are treated in \cite{PJ10}.

The set of natural integers is written $\NN$ and includes $0$.
$\intinter{i}{j}$ denotes the set of all integers from $i$ to $j$.

A \emph{multiset} $m$ over a set $A$ is a pair $(f, A)$ where $f : A \longrightarrow \NN$.
For $a \in A$, $f(a)$ is the number of occurrences of $a$ in $m$.
Its \emph{support} is $\Supp(m) = \{ a \in A : f(a) > 0 \}$.
If the support is finite, its cardinal $|m| = \sum_{a \in A} f(a)$ is finite as well.
We write $a \in m$ when $a \in \Supp(m)$.
The sum $m_0 + m_1$ of two multisets $m_0 = (f_0,A)$ and $m_1=(f_1,A)$ is the multiset $m = (f,A)$ where $f(a) = f_0(a) + f_1(a)$ for all $a \in A$.
$m_0$ is included in $m_1$, written $m_0 \sqsubseteq m_1$, when $f_0(a) \leqslant f_1(a)$ for all $a \in A$.

\subsection{Symport/Antiport P Systems}

\begin{definition}
    A \emph{membrane structure} is a rooted tree in which nodes are labeled by integers from $1$ to $q$, where $q > 0$ is the size of the membrane structure.
    The nodes are called membranes.
    The depth $d(h)$ of a membrane $h$ is its distance from the root.
    The root is called the \emph{skin}.
    A leaf is called an \emph{elementary} membrane.
    The parent of a membrane $h$ is denoted by $p(h)$.
\end{definition}

\begin{definition}
    \label{definition:symportantiportpsystem}
    A \emph{symport/antiport P system with membrane separation rules} is a tuple
    $$\Pi = (\Gamma, (\Gamma_0, \Gamma_1), \Sigma, \mathcal{E}, \mu, \mathcal{M}_1, \dots, \mathcal{M}_q, \mathcal{R}_1, \dots, \mathcal{R}_q, h_{in}, h_{out})$$
    where:
    \begin{itemize}
    \item $\Gamma$ is a finite alphabet of objects with a partition $\{\Gamma_0, \Gamma_1\}$: $\Gamma_0 \uplus \Gamma_1 = \Gamma$,
    \item $\Sigma \subset \Gamma$ is the input alphabet,
    \item $\mathcal{E} \subseteq \Gamma \setminus \Sigma$ is the set of objects present in the environment, each one available in arbitrary number of copies,
    \item $\mu$ is a membrane structure of size $q$,
    \item $\mathcal{M}_h$ is the finite multiset of objects of $\Gamma \setminus \Sigma$ initially located in membrane $h \in \intinter{1}{q}$,
    \item $\mathcal{R}_h$ is the finite set of evolution rules of the membrane $h \in \intinter{1}{q}$ of type:
    \begin{itemize}
        \item $\sympout{u}$ or $\sympin{u}$ called the symport rules,
        \item $\antiport{u}{v}$ called the antiport rules,
        \item $[a]_h \rightarrow [\Gamma_0]_h[\Gamma_1]_h$, if $h$ is an elementary membrane distinct from the skin and the output membrane $h_{out}$, called the separation rules,
    \end{itemize}
    where $u$ and $v$ are non-empty finite multisets of $\Gamma$ and $a \in \Gamma$,
    \item $h_{in} \in \intinter{1}{q}$ is the index of the input membrane,
    \item $h_{out} \in \intinter{0}{q}$ is the index of the output membrane.
    \end{itemize}
\end{definition}

\paragraph{How Does a Symport/Antiport P System Compute?}
Let $\Pi$ be a symport/ antiport P system with membrane separation rules defined as above.
As P systems are models of computation, they evolve step by step from an initial configuration.
Our definition allows to give a multiset $m$ over the input alphabet $\Sigma$ as an input to the system in the membrane $h_{in}$.
We describe here the evolution of $\Pi(m)$.
Note that the computation is non-deterministic, but we will later restrict ourselves to \emph{confluent} systems, where all possible computations give the same result (see Remark \ref{remark:confluence}).

Each membrane $h \in \mu$ of the system delimits a \emph{region} in which objects of $\Gamma$ lie.
A region is the area between $h$ and its children (membranes $h'$ such that $p(h') = h$).
Objects go through membrane $h$, when communicated between region $h$ and $p(h)$, according to symport and antiport rules in $\mathcal{R}_h$.
Elementary membranes (the one having no children) can split according to the separation rules.
We distinguish two regions: the \emph{environment} is the region outside the system, containing in arbitrary many occurrences the objects of $\mathcal{E}$.
The \emph{skin} is the limit between the system and the environment, its index is the root of $\mu$.
It cannot separate, however, it is the only membrane able to communicate with the environment, so we consider the environment as the parent of the skin.

The \emph{configuration} $\mathcal{C}_n = (\mu_n, (m_n(h))_{h\in \mu_n})$ of the system after step $n \in \NN$ is a pair where $\mu_n$ is the membrane structure of the system and $(m_n(h))_{h\in \mu_n}$ is the family of all multisets corresponding to the objects lying in each region after step $n$.
The \emph{initial configuration} is $\mathcal{C}_0 = (\mu_0, (m_0(h))_{h\in \mu_0})$ where $\mu_0 = \mu$, $m_0(h) = \mathcal{M}_{h}$ for all $h \in \intinter{1}{q}$, $h \neq h_{in}$, and $m_0(h_{in}) = \mathcal{M}_{h_{in}} + m$.

The symport/antiport P system evolves according to the \emph{maximal parallel mode}, that is: all rules that can be applied are applied in parallel as long as they can.
Objects can be used by only one rule at a given step, and the objects communicated can't be used again in the same step.
If an object can be used by more than one rule, the applied rule is chosen non-deterministically.
If a separation rule is applied, it is the only one applied to this membrane.

During step $n>0$, the symport rule $\sympout{u} \in \mathcal{R}_h$ can be used if the objects of $u$ are in region $h$: $u \sqsubseteq m_{n-1}(h)$.
It sends the objects of $u$ outside region $h$ (through membrane $h$ to the region $p(h)$).
The symport rule $\sympin{u} \in \mathcal{R}_h$ can be used if $u \sqsubseteq m_{n-1}(p(h))$.
It sends $u$ from region $p(h)$ inside region $h$.
The antiport rule $\antiport{u}{v} \in \mathcal{R}_h$ can be used if $u \sqsubseteq m_{n-1}(h)$ and $v \sqsubseteq m_{n-1}(p(h))$.
Then $u$ is sent outside $h$ whereas $v$ is sent inside $h$.

The length of a rule is $|u|$ for symport rules and $|u| + |v|$ for antiport rules.

The membrane separation rule $[a]_h \rightarrow [\Gamma_0]_h[\Gamma_1]_h \in \mathcal{R}_h$ can be used when the object $a$ is in $h$.
The object $a$ is consumed, and the membrane $h$ is split into two membranes, the first one containing the objects of $\Gamma_1$ in $h$ and the second the objects of $\Gamma_2$ in $h$.
The initial membrane $h$ is deleted and the new membranes keep the index $h$ so the rules of $\mathcal{R}_h$ can be applied to both at the next step.
By definition, only elementary membranes different from the output membrane can be separated.

Starting with the initial configuration $\mathcal{C}_0$ and applying the rules in the maximal parallel order, we obtain a sequence $\mathcal{C}_0, \mathcal{C}_1, \dots$ of configurations called computation, representing the state of the system after each step.
When no rules can be applied to a configuration, the computation halts.
In this case, the result of the computation is the multiset of objects in the membrane $h_{out}$ of the final configuration.
If a computation never halts, then it produces no output.

\subsection{Decision Problems and Complexity Classes}

\begin{definition}
A \emph{decision problem} $X$ is a pair $(I_X, \theta_X)$ where $I_X$ is the set of instances and $\theta_X : I_X \longrightarrow \{0, 1\}$.
An instance $x \in I_X$ is said to be accepted by the problem if $\theta_X(x) = 1$, else it is rejected.
\end{definition}

\begin{definition}
    \label{definition:recognizer}
    A \emph{recognizer symport/antiport P system with membrane separation rules} is a symport/antiport P system with membrane separation $\Pi$ where:
    \begin{itemize}
        \item $\mathbf{yes}$, $\mathbf{no} \in \Gamma \setminus \Sigma$ are the answer objects,
        \item they are initially present in the system: $\exists h, h' \in \intinter{1}{q}, \mathbf{yes} \in m_h, \mathbf{no} \in m_{h'}$,
        \item the output membrane is the environment: $h_{out} = 0$,
    \end{itemize}
    and such that:
    \begin{itemize}
        \item every computation halts,
        \item for all computations of $\Pi$ on the same input, either $\mathbf{yes}$ (the computation is said to be accepting) or $\mathbf{no}$ (the computation is said to be rejecting) must be sent to the environment \textbf{only at the last step}.
    \end{itemize}
    We denote by $\CSC$ the class of all recognizer symport/antiport P systems with separation rules.
    When all rules have length at most $k \in \NN$, the recognizer P system is in $\CSCk{k}$.
\end{definition}

\begin{definition}
    Let $\mathcal{D}$ be a class of recognizer P systems.
    A family $\mathbf{\Pi} = (\Pi_n)_{n \in \NN}$ of P systems is \emph{$\mathcal{D}$-consistent} if $\Pi_n \in \mathcal{D}$ for all $n \in \NN$.
    Moreover, it is said to be \emph{polynomially uniform} if there exists a deterministic Turing Machine that computes, from the unary representation of $n \in \NN$, the P system $\Pi_n$ in polynomial time of $n$.
\end{definition}

\begin{definition}
    A decision problem $X = (I_X, \theta_X)$ is \emph{decidable in a uniform way} by a polynomially uniform family $\mathbf{\Pi}$ of recognizer P systems if there exist two polynomially computable functions $s$ (size) and $cod$ (encoding) such that for all $x \in I_X$:
    \begin{itemize}
        \item $s(x) \in \NN$,
        \item $cod(x)$ is a multiset over $\Sigma_{s(x)}$, the input alphabet of $\Pi_{s(x)} \in \mathbf{\Pi}$,
        \item if $\theta_X(x) = 1$ then all computations of $\Pi_{s(x)}(cod(x))$ are accepting,
        \item if there exists one accepting computation of $\Pi_{s(x)}(cod(x))$, then $\theta_X(x) = 1$.
    \end{itemize}
\end{definition}

\begin{remark}
    \label{remark:confluence}
    The last two properties are called \emph{soundness} and \emph{completeness}.
    They imply \emph{confluence}, which is that the computations of $\Pi_{s(x)}(cod(x))$, $x\in I_X$ are all accepting or all rejecting.
    Such a P system is said to be confluent.
\end{remark}

\begin{definition}
    Let $\mathcal{D}$ be a class of recognizer P systems.
    A decision problem $X$ is in $\mathbf{PMC}_\mathcal{D}$ if there exists a polynomially uniform $\mathcal{D}$-consistent family $\mathbf{\Pi}$ deciding $X$ in a uniform way, such that for all $x \in I_X$, each computation of $\Pi_{s(x)}(cod(x))$ halts in a polynomial number of steps relatively to $s(x)$.
\end{definition}

\begin{proposition}
    \label{proposition:pmc}
    Let $\mathcal{D}$, $\mathcal{D'}$ be two classes of recognizer P systems such that $\mathcal{D} \subseteq \mathcal{D'}$.
    The following statement holds by definition: $\PMC_\mathcal{D} \subseteq \PMC_\mathcal{D'}$.
\end{proposition}

\section{\texorpdfstring{Characterization of $\PsharpP$}{Characterization of P\textasciicircum(\#P)}}
\label{section:results}

The main contribution of this article is the proof that the class $\PsharpP$ characterizes exactly the problems solved by symport/antiport P systems with membrane separation rules, see Corollary \ref{corollary:characterization}.

This section provides all the tools and the results needed in the proof of this characterization.
First, we recall a previous result from \cite{LepoAl17} characterizing $\PsharpP$ for tissue-like symport/antiport P systems, another type of systems where membranes are organized in a graph structure (instead, we use a more specific definition with a hierarchy structure between the membranes).
Then we state our principal theorem: we can solve the $\PsharpP$-complete problem $\MIDSAT$ with a symport/antiport P system with membrane separation using rules of length not exceeding 3 (Theorem \ref{theorem:midsat_solved}).

Combining these results, we prove the characterization (Corollary \ref{corollary:characterization}).

\begin{definition}
    A \emph{tissue P system} is a tuple
    $$\Pi = (\Gamma, (\Gamma_0, \Gamma_1), \Sigma, \mathcal{E}, \mathcal{M}_1, \dots, \mathcal{M}_q, \mathcal{R}, h_{in}, h_{out})$$
    The difference with the cell-like P systems we defined (Definition \ref{definition:symportantiportpsystem}) is that we do not have a membrane structure anymore.
    Every membrane can separate (except the output membrane) and can communicate objects by symport/antiport rules to any other membrane, including the environment.
    The rules are all grouped in the rule set $\mathcal{R}$.
    Symport rules are written $[u]_h \leftrightarrow [\ ]_{h'}$, and antiport rules are written $[u]_h \leftrightarrow [v]_{h'}$ ($u$, $v$ being non-empty finite multisets over $\Gamma$; $h$, $h' \in \intinter{0}{q}$ being the communicating membranes, where $0$ identifies the environment).

    To define $\TSC$, the class of \emph{recognizer tissue symport/antiport P systems with membrane separation}, we ask for the same properties as for the cell-like systems ($\mathbf{yes}$ and $\mathbf{no}$ objects, $h_{out} = 0$, computations halt when releasing $\mathbf{yes}$ or $\mathbf{no}$ in the environment, see definition \ref{definition:recognizer}).
\end{definition}

\begin{proposition}
    \label{proposition:cell_included_tissue}
    $\PMCCSC \subseteq \PMCTSC$
\end{proposition}

\begin{proof}
    A recognizer (cell-like) P system is, in fact, a recognizer tissue P system.
    We just need to translate symport and antiport rules in $\mathcal{R}_h$ as follows: $\sympout{u}$, $\sympin{u}$, and $\antiport{u}{v}$ become $[u]_h \leftrightarrow [\ ]_{p(h)}$, $[u]_{p(h)} \leftrightarrow [\ ]_h$, and $[u]_h \leftrightarrow [v]_{p(h)}$.

    Then $\CSC \subseteq \TSC$ and $\PMCCSC \subseteq \PMCTSC$ by Proposition \ref{proposition:pmc}.
\end{proof}

\begin{proposition}[Leporati et Al., \cite{LepoAl17}]
    \label{proposition:characterization_tissue}
    $\PMCTSC = \PsharpP$
\end{proposition}

\begin{definition}
    The $\MIDSAT$ decision problem is defined as follows:

    \textbf{Instance:} A satisfiable boolean formula $\phi(x_1, \dots, x_n)$ in conjuntive normal form.
    
    \textbf{Question:} Does the least significant bit of the middle (according to lexicographic order) satisfying assignment of $\phi$ equal $1$?
\end{definition}

The question is to determine the truth value of $x_n$ in the assignment which is the median among all the satisfying assignments of $\phi$, ordered by the lexicographic order.
This is one of the problems complete for $\PsharpP$ listed by Toda in \cite{Toda94}.

\begin{proposition}[Toda, \cite{Toda94}]
    \label{proposition:midsat_complete}
    $\MIDSAT$ is $\PsharpP$-complete.
\end{proposition}

\begin{example}
    With the boolean formula $\phi(x_1,x_2,x_3) = (x_1 \lor x_2 \lor \lnot x_3) \land (x_1 \lor \lnot x_2 \lor x_3) \land (\lnot x_1 \lor x_2 \lor x_3)$, the satisfying assignments are $(0,0,0)$, $(0,1,1)$, $(1,0,1)$, $(1,1,0)$, $(1,1,1)$, ordered in the lexicographic order.
    Then the median is $(1,0,\mathbf{1})$ and the answer to the problem is \textbf{yes}.
    
    When an even number of assignments satisfy the formula, the lower of the two median assignments is chosen as the median.
\end{example}

\begin{theorem}
\label{theorem:midsat_solved}
    $\MIDSAT \in \PMCCSCk{3}$
\end{theorem}

We prove this result in the following section by implementing a family of recognizer symport/antiport P systems with membrane separation rules and rules of length at most 3 solving $\MIDSAT$ in a uniform way and in polynomial time.

\begin{corollary}[Characterization of $\PsharpP$]
    \label{corollary:characterization}
    $\PMCCSCk{3} = \PsharpP$
\end{corollary}

\begin{proof}
    By definition, $\PMCCSCk{3} \subseteq \PMCCSC$, and from Propositions \ref{proposition:cell_included_tissue} and \ref{proposition:characterization_tissue}, we obtain $\PMCCSCk{3} \subseteq \PsharpP$.
    
    By Theorem \ref{theorem:midsat_solved} and because $\MIDSAT$ is $\PsharpP$-complete (Proposition \ref{proposition:midsat_complete}), the following inclusion holds: $\PsharpP \subseteq \PMCCSCk{3}$.

    The corollary comes out from these inclusions.
\end{proof}

\section{\texorpdfstring{A System Solving $\MIDSAT$}{A System Solving MIDSAT}}
\label{section:system}

This section is devoted to the proof of Theorem \ref{theorem:midsat_solved}.

The P system we create is an evolution of the one used to solve the $\SAT$ problem in \cite{VCAl15}.
We first explain briefly how their system works and what changes we made.
Then we provide an overview of the computation.
Finally, we give a detailed proof of the theorem.

The complexity of the system arises principally from the constraint to have rules of length at most three.
This requires us to add steps when we want four objects to interact at a precise step or when we want to exchange an object with $n$ copies of another one.

\subsection{\texorpdfstring{Some Modifications to the $\CSCk{3}$ $\SAT$ Solver}{Some Modifications to the CSC(3) SAT Solver}}

Our work is based on the system implemented by \cite{VCAl15} solving $\SAT$.

This one runs in three phases: the generation phase, where $2^n$ membranes are created with the separation rule, each one representing an assignment of the $n$ variables of the formula; the checking phase consists in the evaluation of the formula in each membrane; finally, the output phase answers $\mathbf{yes}$ or $\mathbf{no}$ depending on whether a membrane satisfies the formula.

The instance is a boolean formula $\phi(x_1, \dots, x_n) = \bigwedge_{j = 1}^m C_j$ in CNF, where $C_j = l_{i_1,j} \lor \cdots \lor l_{i_{r_j},j}$ and $l_{i_k,j} = x_{i_k}$ or $\lnot x_{i_k}$, $j \in \intinter{1}{m}$, $k \in \intinter{1}{r_j}$.
They encode it by $cod(\phi) = \{ x_{i_k,j} : l_{i_k,j} = x_{i_k,j} \} + \{ \bar{x}_{i_k,j} : l_{i_k,j} = \lnot x_{i_k,j} \}$, and $s(\phi) = \langle m,n \rangle$, where $\langle m,n \rangle = ((n + m)(n + m + 1)/2) + n$ is the common bijection between $\NN$ and $\NN^2$.

We denote by $\Pi^{\mathbf{SAT}}_{\langle m, n \rangle}$ the systems of the family they provide.
$$\Pi^{\mathbf{SAT}}_{\langle m, n \rangle} = (\Gamma, (\Gamma \setminus \Gamma_1, \Gamma_1), \Sigma, \mathcal{E}, [ [\ ]_2 [\ ]_3 ]_1, \mathcal{M}_1, \mathcal{M}_2, \mathcal{M}_3, \mathcal{R}_1, \mathcal{R}_2, \mathcal{R}_3, 1, 0)$$

The system of \cite{VCAl15} has initially three membranes.
The skin (membrane 1) is the input membrane.
It contains: membrane 2, that will separate to obtain $2^n$ membranes 2 (one for each truth assignment), and membrane 3, used during the output phase.
The rules have a maximum length of 3 objects.

At the end of the checking phase, that is after $\DeltaSAT = 3n + 2m$ steps, every membrane 2 corresponding to a satisfying assignment of $\phi$ contains an object $e_{i,m}$ or $\bar{e}_{i,m}$ for $i \in \intinter{1}{n}$.

In our implementation, we re-use this system until the end of the checking phase.
We delete the membrane 3 and the rules associated with it, because it was used only during the output phase.
We remove the objects $\{ f_r \}_{r \in \intinter{0}{\DeltaSAT}} \cup \{ f'_p \}_{p \in \intinter{0}{\DeltaSAT+1}} = D$.
We remove also the rules $\{\sympout{E_0 f_{\DeltaSAT} \mathbf{yes}}, \sympout{f_{\DeltaSAT} \mathbf{no}}\} = D_1 \subseteq \mathcal{R}_1$, and $\{\sympout{e_{i,m}E_0}, \sympout{\bar{e}_{i,m}E_0}\}_{i \in \intinter{1}{n}} = D_2 \subseteq \mathcal{R}_2$, so that at this point, we obtain a halting configuration after $\DeltaSAT$ steps.

Then, we add three phases so the computation solves $\MIDSAT$:
\begin{itemize}
    \item The \emph{activation phase} makes the membranes 2 corresponding to a satisfying assignment able to interact during the next phases.
    \item The \emph{search phase} searches for the lexicographically middle satisfying assignment, determining the values of the variables, one by one.
    The determination of the $k$-th variable needs three sub-phases ($k \in \intinter{1}{n}$):
    \begin{itemize}
        \item \emph{Preparation sub-phase}: the preparation cannot be done previously in parallel because it needs the result of the previous
        determination phase.
        \item \emph{Comparison sub-phase}: all the satisfying assignments are compared to a well-chosen assignment, see the next paragraph giving the main idea.
        \item \emph{Determination sub-phase}: depending on the number of assignments lower or greater than the one they have been compared with, the $k$-th variable of the median assignment is determined.
    \end{itemize}
    \item The \emph{output phase} returns $\mathbf{yes}$ or $\mathbf{no}$ whether $x_n = 1$ or $x_n = 0$ in the assignment found previously.
\end{itemize}

We describe all these phases more precisely in the following subsection.

\paragraph{The Key Idea.}
To determine the median $x_1 \cdots x_n$ of satisfying assignments, the search phase does the same thing as a research of the median of a set of values between $0$ and $2^n-1$.
This works by finding the bits of the median, from the most to the least significant bit, by counting the number of values greater and lower than a current well-chosen number.
This process is very close to a binary search: we first consider $10 \cdots 0 = 2^{n-1}$, if there are strictly more numbers greater or equal than $10\cdots 0$, than numbers less than $10\cdots 0$, then the most significant bit $x_1$ of the median is $1$, else $0$.
To find the $k$-th most significant bit $x_k$, we must consider $x_1 \cdots x_{k-1} 1 0 \cdots 0$.
After $n$ iterations, we obtain the median.

For the sake of a better understanding of the different phases, we define a few constants for the number of steps each (sub-)phase lasts for and when it begins:

\begin{center}
    \begin{tabular}{|l|l|}
        \hline
        \multicolumn{1}{|c|}{Sub-phases duration} & \multicolumn{1}{|c|}{Phases duration} \\
        \hline
        $\DeltaPrep = n+5$ & $\DeltaSAT = 3n + 2m$ \\
        \hline
        $\DeltaComp = 2n+1$ & $\DeltaACT = n+2$ \\
        \hline
        $\DeltaDete = 3$ & $\DeltaSEA = n \cdot \DeltaLoop$ \\
        \hline
        $\DeltaLoop = \DeltaPrep + \DeltaComp + \DeltaDete$ & $\DeltaOUT = 1$ \\
        \hline
        \multicolumn{2}{|c|}{Initial step of (sub)-phases} \\
        \hline
        \multicolumn{2}{|l|}{$\IPrep(k) = \DeltaSAT + \DeltaACT + (k-1) \cdot \DeltaLoop$, $k \in \intinter{1}{n}$} \\
        \hline
        \multicolumn{2}{|l|}{$\IComp(k) = \IPrep(k) + \DeltaPrep$, $k \in \intinter{1}{n}$} \\
        \hline
        \multicolumn{2}{|l|}{$\IDete(k) = \IComp(k) + \DeltaComp$, $k \in \intinter{1}{n}$} \\
        \hline
        \multicolumn{2}{|l|}{$\IOUT = \DeltaSAT + \DeltaACT + \DeltaSEA$} \\
        \hline
    \end{tabular}
\end{center}

Thus, the family $\mathbf{\Pi} = (\Pi_{\langle m, n\rangle})_{m, n}$ we provide (described earlier), extends the system $\Pi^{\mathbf{SAT}}_{\langle m, n \rangle}$.
It is a polynomially uniform family of recognizer symport/antiport P systems with membrane separation and rules of length at most 3, defined as follows:
$${\Pi}_{\langle m, n \rangle} = (\Gamma', (\Gamma' \setminus \Gamma_1, \Gamma_1), \Sigma, \mathcal{E}', [ [\ ]_2 ]_1, \mathcal{M}'_1, \mathcal{M}_2, \mathcal{R}'_1, \mathcal{R}'_2, 1, 0)$$
\begin{itemize}
    \item $\mathcal{R}'_1 = (\mathcal{R}_1 \setminus D_1) \cup \mathcal{R}''_1$,
    \item $\mathcal{R}'_2 = (\mathcal{R}_2 \setminus D_2) \cup \mathcal{R}''_2$,
    \item $\mathcal{E}' = (\mathcal{E} \setminus D) \cup \mathcal{E}''$,
    \item $\Gamma' = (\Gamma \setminus D) \cup \mathcal{E}'' \cup \{ \xi_{i,0} \}_{i \in \intinter{1}{n}} \cup \{ \omega_0 \}$,
    \item $\mathcal{M}'_1 = \mathcal{M}_1 - \{ f_0 \} - \{ f'_p \}_{p \in \intinter{1}{\DeltaSAT+1}} + \{ \mathbf{no}, \omega_0 \} + \{ \xi_{i,0} \}_{i \in \intinter{1}{n}}$.
\end{itemize}

The reader could find the alphabet of the objects added to the environment $\mathcal{E}''$ and the lists $\mathcal{R}''_1$ and $\mathcal{R}''_2$ of rules added to membranes $1$ and $2$ in the Appendix \ref{appendix:definition}.

\subsection{Computation Overview}

Here we present a detailed overview of the computation of the system $\Pi_{\langle m, n \rangle}$ on input $cod(\phi)$.
As mentioned in the previous subsection, the system solves $\SAT$ until step $\DeltaSAT = 3n +2m$, where the membranes 2 corresponding to a satisfying assignment of $\phi$ contain an object $e_{i,m}$ or $\bar{e}_{i,m}$, $i \in \intinter{1}{n}$. 
We explain the effect of activation, search, and output phases.

Remind that the step $s$ applies between configurations $\mathcal{C}_{s-1}$ and $\mathcal{C}_s$.

\paragraph{A Few Words About Counters.}
To control the timing of the system, we use counters situated in membrane $1$.
These objects keep the information of which is the current step and allow some rules to be used at a precise step.
Their objective is to bring other objects into the system.
The counters are the objects $\xi_{i,s}$ and $\omega_{s}$, they appear only in rules of $\mathcal{R}'_1$.

Counters increment their index each step with rules like $\antiport{\xi_{i,t}}{\xi_{i,t+1}}$ and $\antiport{\omega_{t}}{\omega_{t+1}}$.
They verify the following property: after step $s$ (that is, in configuration $\mathcal{C}_s$), the only $\xi$ and $\omega$ objects present in the system are $\xi_{i,s}$, $\omega_s \in m_s(1)$, in membrane $1$, $i \in \intinter{1}{n}$.

We have $n$ variants of the $\xi$ object thanks to the index $i \in \intinter{1}{n}$.
Sometimes, they bring together with their successor $\xi_{i,s+1}$ a different object so that, at a precise step, we can make appear $n$ (possibly distinct) objects in membrane $1$:
\begin{itemize}
    \item $\nu_{i,0}$, at step $\DeltaSAT +1$, $i \in \intinter{1}{n}$,
    \item $\zeta'_{k,i}$, at step $\IPrep(k)$, $k \in \intinter{1}{n}$, $i \in \intinter{1}{n}$,
    \item $\bar{\mu}_{k}$, at step $\IPrep(k+1) +1$, $k \in \intinter{1}{n}$.
\end{itemize}

$\omega$ objects are present in the system in $2^n$ ocurrences.
They are useful to bring into the system a lot of objects in only one step.
The rules $\antiport{\omega_{s}}{\omega_{s+1}^2}$, $s \in \intinter{0}{n-1}$ apply in the first $n$ steps of the computation (in parallel to $\SAT$ computation) to duplicated $\omega$ objects until they are in $2^n$ occurrences.
Then, as for $\xi$, $\omega$ counters bring the following objects in the system within $2^n$ occurrences:
\begin{itemize}
    \item $\omega'_{\DeltaSAT -1}$, at step $\DeltaSAT -1$,
    \item $\sigma$, from step $\DeltaSAT +1$ to step $\DeltaSAT +n$,
    \item $\theta_{\IComp(k)-1}$, at step $\IComp(k) -1$, $k \in \intinter{1}{n}$,
    \item $\omega'_{\IComp(k) + 2i}$, at step $\IComp(k) + 2i$, $k \in \intinter{1}{n}$, $i \in \intinter{0}{n-2}$,
    \item $p'_{k,i+1}$, at step $\IComp(k) + 2i +1$, $k \in \intinter{1}{n}$, $i \in \intinter{0}{n-2}$,
    \item $\kappa'''_{k}$, at step $\IComp(k) + 2n$, $k \in \intinter{1}{n}$,
    \item $\chi_{k}$, at step $\IDete(k)$, $k \in \intinter{1}{n}$.
\end{itemize}

\paragraph{Activation Phase.}
This phase lasts for $\DeltaACT = n + 2$ steps and brings in the membranes $2$ corresponding to a satisfying assignment of $\phi$ the objects $\Lambda_i$ (the assignment makes $x_i$ true) and $\Psi_i$ (the assignment makes $x_i$ false) that will interact in the next phases ($i \in \intinter{1}{n}$).

In membrane 1, during steps $\DeltaSAT +s$, $s \in \intinter{0}{n-2}$, objects $\omega'_{\DeltaSAT +s}$ bring the objects $\sigma'_s$ in $2^n$ occurrences, and objects $\nu_{i,s}$ duplicated.
At step $\DeltaSAT + n -1$, $2^n$ objects $\sigma'_{n-1}$ enter the system and every $\nu_{i,n-1}$ ($i \in \intinter{1}{n}$, in $2^n$ occurrences) bring objects $\Lambda_i$ and $\Psi_i$ from the environment.
In parallel, from steps $\DeltaSAT+1$ to $\DeltaSAT+n$, $2^n$ objects $\sigma$ enter membrane $1$ with the help of the counter $\omega$.

At step $\DeltaSAT +1$, each membrane $2$ corresponding to a satisfying assignment of $\phi$ exchanges its object $e_{i,m}$ (or $\bar{e}_{i,m}$) for $\sigma'_0$ and $T_i$ (or $F_i$).
Then during each step $\DeltaSAT + i+2$, $\sigma'_i$ is exchanged with $\sigma'_{i+1}$ and one object $\sigma$ (for $i\in\intinter{0}{n-2}$), and at step $\DeltaSAT + n+1$, an other object $\sigma$ enters those membranes $2$.
Finally, at step $\DeltaSAT + n+2 = \DeltaSAT + \DeltaACT$, the rules $\antiport{\sigma, T_{i}}{\Lambda_{i}}$ and $\antiport{\sigma, F_{i}}{\Psi_{i}} \in \mathcal{R}'_2$ are used so the objects $T_i$, $F_i$ are replaced by $\Lambda_i$ and $\Psi_i$ in the activated membranes.

\paragraph{Search Phase.}
This phase determines the values of the variables $x_1, \dots, x_n$ of the lexicographically middle satisfying assignment.
The P system will repeat $n$ times the three sub-phases, finding the value of one variable at each loop.
The variables are determined in order, from $x_1$ to $x_n$.

Let's describe the $k$-th iteration of the sub-phases, in which $x_k$ is determined ($k \in \intinter{1}{n}$).
Preparation, comparison, and determination sub-phases last for $\DeltaPrep = n + 5$, $\DeltaComp = 2n + 1$, and $\DeltaDete = 3$ steps.
Thus, the search phase lasts for $n \cdot \DeltaLoop= 3n^2 + 9n$ steps.

\paragraph{Preparation Sub-phase $k$.}
When this sub-phase begins, the values $x_1, \dots, x_{k-1}$ are coded by the presence in membrane 1 of $\mu_j$, if $x_j = 1$, or $\bar{\mu}_j$, if $x_j = 0$ ($j \in \intinter{1}{k-1}$).
The objective of this phase is to bring $2^n$ occurrences of the objects $\ddot{\epsilon}^N_{k,1} \cdots \ddot{\epsilon}^N_{k,k-1}$ $\epsilon_{k,k}$ $\bar{\epsilon}^N_{k,k+1} \cdots \bar{\epsilon}^N_{k,n}$, where $\ddot{\epsilon}^N_{k,j} = \epsilon^N_{k,j}$ or $\bar{\epsilon}^N_{k,j}$ whether $x_j = 1$ or $0$, $j \in \intinter{1}{k-1}$.

Here, we use $\epsilon$ for $l$, $g$, $p$; and $N$ for $\Lambda$, $\Psi$, so that in fact $6n \cdot 2^n$ objects must appear (6 possible combinations for each $i \in \intinter{1}{n}$, in $2^n$ occurrences).
The presence/absence of the bar on these objects represents $x_1 \cdots x_{k-1} 1 0 \cdots 0$, the reference assignment we want to compare to all satisfying assignments.
Moreover, $n2^n$ occurrences of $\phi$, $\phi'$ will appear in the system at the end of this step.

This sub-phase uses only rules from the rule set $\mathcal{R}'_1$ of membrane $1$.
One step before the preparation sub-phase, objects $\zeta'_{k,i}$ are created with the $\xi_{i,\IPrep(k)}$ counters.
Next step is $\IPrep(k) +1$, and $\zeta'_{k,i}$ objects become either $\zeta_{k,i}$ or $\bar{\zeta}_{k,i}$, with a bar iff the $i$-th value of the reference is a zero.
For $i < k$, the choice is done by using $\mu_i$ or $\bar{\mu}_i$.
For $i = k$, it will be $\zeta_{k,k}$, and for $i > k$, $\bar{\zeta}_{k,i}$.
Then $\zeta_{k,i}$/$\bar{\zeta}_{k,i}$ are exchanged for $\eta_{k,i,0}$/$\bar{\eta}_{k,i,0}$, and if $i < k$ then $\mu_i$/$\bar{\mu}_i$ re-enter membrane $1$.
Objects $\eta_{k,i,j}$/$\bar{\eta}_{k,i,j}$ are duplicated at each of the $n$ next steps, until $\eta_{k,i,n}$/$\bar{\eta}_{k,i,n}$ are present in membrane $1$ in $2^n$ occurrences.
Objects $\eta_{k,i,n}$/$\bar{\eta}_{k,i,n}$ are exchanged with $\lambda_{k,i}$/$\bar{\lambda}_{k,i}$ and $\pi_{k,i}$/$\bar{\pi}_{k,i}$, the formers then become $l_{k,i}$/$\bar{l}_{k,i}$ and $g_{k,i}$/$\bar{g}_{k,i}$, whereas the laters become $p_{k,i}$/$\bar{p}_{k,i}$ and $\pi'_k$.
Finally, objects $\epsilon^N_{k,i}$/$\bar{\epsilon}^N_{k,i}$, appear in $2^n$ occurrences ($N \in \{\Lambda, \Psi\}$, $i\in\intinter{1}{n}$) and objects $\phi$, $\phi'$ appear in $n2^n$ occurrences.

\paragraph{Comparison Sub-phase $k$.}
During this phase, every active membrane 2 will compare their assignment $N_1\cdots N_n$ (stored by the objects $\Lambda_i$/$\Psi_i$, $i \in \intinter{1}{n}$) to the reference assignment $\nu = x_1 \cdots x_{k-1} 1 0 \cdots 0$ coded in membrane 1.
An object $\epsilon^N_{k,i}$/$\bar{\epsilon}^N_{k,i}$ is sent in the membrane 2 to indicate if the number given by the first $i$ bits of the assignment is lower ($\epsilon = l$), equal ($\epsilon = p$) or greater ($\epsilon = g$) than the one given by the first $i$ bits of the reference.

The following property holds during the comparison sub-phase:

For $i\in\intinter{1}{n-1}$, the object $\epsilon^N_{k,i}$/$\bar{\epsilon}^N_{k,i}$ present in a membrane $2$ after step $\IComp(k) + 2i-1$ verifies: (1) $\epsilon = l$ if $N_1\cdots N_i < \nu_1\cdots\nu_i$, $\epsilon = p$ if $N_1\cdots N_i = \nu_1\cdots\nu_i$, or $\epsilon = g$ if $N_1\cdots N_i > \nu_1\cdots\nu_i$; (2) $N = \Lambda$ if $N_i = \Lambda_i$ or $N = \Psi$ if $N_i = \Psi_i$; and (3) there is a bar over $\epsilon$ iff $\nu_i = 0$.
For the case $i=n$, the property asks to consider configuration $\IComp(k) + 2n$ instead of $\IComp(k) + 2n -1$.

The exponent $N = \Lambda$ or $\Psi$ is used to know which object among $\Lambda_i$ and $\Psi_i$ must be re-sent to the membrane 2 in the next step.

First, we explain what happens between the environment and membrane $1$ during this sub-phase.
At step $\IComp(k) +2i +1$, $i\in\intinter{0}{n-2}$, object $\omega'_{\IComp(k) +2i}$ that barely appeared in membrane $1$, will be exchanged with objects $l'_{k,i+1}$ and $g'_{k,i+1}$.
At the same time, object $p'_{k,i+1}$ enters the system.
These three objects are present in $2^n$ occurrences.
Still during this first step, object $\phi'_k$ becomes $\phi''_k$.
In the next step, objects $\phi_k$ and $\phi''_k$ will leave the system together.
Independently, object $\theta_{\IComp(k) + j}$ (created before by the counter $\omega$) will duplicate until $\theta_{\IComp(k) + 2n -2}$ being in $6\cdot 2^n$ occurrences in configuration $\mathcal{C}_{\IComp(k) + 2n -2}$.
Then it is exchanged with $\kappa_k$ and $\kappa'_k$.
At step $\IComp(k) + 2n$, $\kappa'_k$ becomes $\kappa''_k$ and some occurrences of $\kappa_k$ leave the system together with all objects $\epsilon^N_{k,n}$/$\bar{\epsilon}^N_{k,n}$ present in membrane $1$.
Such objects are still in the system, but in membranes $2$.
To avoid them to be pulled out from the system after they leave membranes $2$, the objects $\kappa'_k$ are sent out by objects $\kappa''_k$.
We recall that $2^n$ objects $\kappa'''_k$ appear in membrane $1$ at step $\IComp(k) + 2n$ with the counter $\omega$.

Secondly, for activated membranes $2$, the comparison sub-phase decomposes into $n-1$ groups of two steps and one group of three steps, each one extending the comparison by one bit.
First (steps $\IComp(k) +1$, $\IComp(k) +2$), an object $\epsilon^N_{k,1}$/$\bar{\epsilon}^N_{k,1}$ enters the membrane together with $\phi_k$, pulling out $N_1$, with respect to the properties (1), (2) and (3).
Then $N_1$ can re-enter the membrane $2$ with $\epsilon'_{k,1}$ to keep the result of the comparison.
Steps $\IComp(k) + 2(i-1)+1$, $\IComp(k) + 2(i-1)+2$, $i\in\intinter{2}{n-1}$ do the same things, but there is no need of $\phi_k$, and they use $\epsilon'_{k,i-1}$ to preserve the properties.
Finally, the group of three steps begins by $\IComp(k) + 2n-1$ which is similar to the others, nothing occurs at step $\IComp(k) + 2n$ (at that time, in membrane $1$, objects $\epsilon^N_{k,n}$/$\bar{\epsilon}^N_{k,n}$ go out of the system).
Last step of the subphase is $\IComp(k) + 2n+1$ where $\epsilon^N_{k,n}$/$\bar{\epsilon}^N_{k,n}$ leaves membrane $2$ and brings back $N_n$ from membrane $1$ to $2$, together with $\kappa'''_{k}$.

The value of $\epsilon$ in this last object keeps the result of the comparison of the reference and the assignment of the membrane $2$ considered.

\paragraph{Determination Sub-phase $k$.}
This sub-phase will count the number of satisfying assignments greater or equal than the reference.
If they are in majority (strictly more than the ones strictly lower than the reference), then the system must contain the object $\mu_k$ at the end of this phase, else it will be $\bar{\mu}_k$.

To count, we first exchange the objects $\bar{g}^{N}_{k,n}$, $\bar{p}^{N}_{k,n}$ (or $g^{N}_{n,n}$, $p^{N}_{n,n}$ when $k=n$) with new objects $G_k$, and the objects $\bar{l}^{N}_{k,n}$ (or $l^{N}_{n,n}$ when $k=n$) with objects $L_k$.
This is done with the help of objects $\chi_k$, created at the previous step with the counter $\omega_{\IDete(k)}$.
At the next step, objects $G_k$ and $L_k$ leave the system by pairs with the rule $\sympout{L_k, G_k} \in \mathcal{R}_1$ and the object $\bar{\mu}_k$ enters membrane $1$ with the counter $\xi_{1,\IDete(k)+2}$.
Finally, if there are objects $G_k$ remaining, the rule $\antiport{G_{k}, \bar{\mu}_{k}}{\mu_{k}}$ applies, bringing $\mu_k$ in the system.
If not, $\bar{\mu}_k$ stays into membrane $1$.

Preparation sub-phase $k+1$ (or the output phase if $k=n$) can now begin.

\paragraph{Output Phase.}
At the end of the search phase (step $\IOUT$), membrane 1 contains $\mu_n$ or $\bar\mu_n$ depending on whether $x_n = 1$ or $0$ in the lexicographically middle satisfying assignment.
This phase lasts for only one step, where one of these two rules $\sympout{\mu_n\ \mathbf{yes}\ \xi_{1, \IOUT}}, \sympout{\bar{\mu}_n\ \mathbf{no}\ \xi_{1, \IOUT}}\in \mathcal{R}_1$ is applied.
The object $\xi_{1, \IOUT}$ has appeared in membrane 1 during the previous step, if it was present previously, the rules could have been applied before.
The objects $\mathbf{yes}$ and $\mathbf{no}$ are present in membrane 1 since the beginning of the computation.

After this step, no rules can be applied yet, and the environment contains either $\mathbf{yes}$ or $\mathbf{no}$, according to the value of $x_n$ in the lexicographically middle satisfying assignment.

\subsection{Proof of Theorem \ref{theorem:midsat_solved}}

We prove here that $\MIDSAT \in \PMCCSCk{3}$.
Previous subsection gives us Lemma \ref{lemma:solved}:

\begin{lemma}
\label{lemma:solved}
    When $x_n = 1$ in the lexicographically middle satisfying assignment of $\phi$, then the object $\mathbf{yes}$ is released in the environment.
\end{lemma}

\begin{proof}[Proof of Theorem \ref{theorem:midsat_solved}]
    First, the family $\mathbf{\Pi} = (\Pi_{\langle m,n \rangle})$ is $\CSCk{3}$-consistent: for all $n, m \in \NN$, $\Pi_{\langle m,n \rangle} \in \CSCk{3}$.

    Secondly, it is polynomially uniform since the alphabet $\Gamma'$ has a polynomial number of objects, initial multisets have a polynomial cardinal, and there is a polynomial number of rules of constant length.

    Then, $\mathbf{\Pi}$ recognizes $\MIDSAT$ in a uniform way: the functions $s(\phi) = \langle m,n \rangle$ and $cod(\phi) = \{ x_{i_k,j} : l_{i_k,j} = x_{i_k,j} \} \cup \{ \bar{x}_{i_k,j} : l_{i_k,j} = \lnot x_{i_k,j} \}$ are polynomially computable, and our system is confluent: the sets of rules we added to the P system solving $\SAT$ are deterministic.
    Lemma \ref{lemma:solved} gives the last argument.

    Finally, every computation of $\Pi_{\langle m,n \rangle}(cod(\phi))$ halts in $\DeltaSAT + \DeltaACT + \DeltaSEA + \DeltaOUT = 3n^2 + 13n + 2m + 3$ steps, a polynomial in $\langle m, n \rangle$.

    Thus, we have $\MIDSAT \in \PMCCSCk{3}$.
\end{proof}

\section{Implementation in P-Lingua}
\label{section:implementation}

We implemented the family $\mathbf{\Pi}$ solving $\MIDSAT$ in the P system simulator \texttt{MeCoSim} \cite{PLingua}, \cite{PLinguaSA}.
The reader will find it on the \href{https://gitlab.inria.fr/vducros/midsat-csc3}{Inria gitlab instance}.

The authors of \cite{VCAl15} have provided the implementation of their system solving $\SAT$ on the \href{http://www.p-lingua.org/mecosim/doc/case_studies/celllike/satcsc.html}{\texttt{MeCoSim} website}.
We based our implementation on it.

During the computation, some objects stay in the system even if they won't be used anymore.
Then it is possible to add garbage rules to make them exit the system.
This has been done in \cite{VCAl15} and we also added such rules.
After some tests, it appears that the simulation is slower when the garbage rules are added, meaning that it is better to let objects lie in the system in a computer simulation.
This has no influence on the number of steps of the computation.

The results of the simulation are gathered in the Table \ref{table:simulation}.

\begin{table}[h]
    \center
    \begin{tabular}{|C{0.08\linewidth}|C{0.2\linewidth}|C{0.2\linewidth}|C{0.18\linewidth}|C{0.18\linewidth}|}
        \hline
        $\langle n, m \rangle$ & Formula $\phi$ & \small Middle satisfying assignment & \small Simulation time with garbage \newline rules (seconds) & \small Simulation time without garbage rules (seconds) \\\hline
        $\langle 3, 3 \rangle$ & 
        \small
        $(x_1 \lor x_2 \lor \lnot x_3)$ \newline
        $\land (x_1 \lor \lnot x_2 \lor x_3)$ \newline
        $\land (\lnot x_1 \lor x_2 \lor x_3)$ & $(1,0,\mathbf{1})$ & $0.990$ & $0.944$ \\\hline
        $\langle 5, 4 \rangle$ & 
        \small
        $(x_1 \lor x_2 \lor \lnot x_5)$ \newline
        $\land (\lnot x_2 \lor x_3)$ \newline
        $\land (\lnot x_1 \lor x_3 \lor x_4)$ \newline
        $\land (x_1 \lor x_4 \lor \lnot x_5)$ & $(1,0,0,1, \mathbf{1})$ & $16.242$ & $15.458$ \\\hline
        $\langle 6, 5 \rangle$ & 
        \small
        $(x_1 \lor x_2 \lor x_5 \lor x_6)$ \newline
        $\land (\lnot x_2 \lor x_3 \lor x_4)$ \newline
        $\land (\lnot x_1 \lor \lnot x_4 \lor x_6)$ \newline
        $\land (x_1 \lor \lnot x_3 \lor \lnot x_5)$ \newline
        $\land (x_3 \lor \lnot x_6)$ & $(1,0,0,0,1,\mathbf{0})$ & $55.108$ & $53.904$ \\\hline
        $\langle 7, 7 \rangle$ &
        \small
        $(x_1 \lor x_2 \lor x_3 \lor x_4)$ \newline
        $\land (x_5 \lor x_6 \lor x_7)$ \newline
        $\land (\lnot x_1 \lor x_2 \lor x_5)$ \newline
        $\land (x_3 \lor \lnot x_5 \lor \lnot x_6)$ \newline
        $\land (x_3 \lor \lnot x_4 \lor x_7)$ \newline
        $\land (\lnot x_1 \lor \lnot x_7)$ \newline
        $\land (\lnot x_2 \lor x_6)$ & $(0,1,0,0,0,1,\mathbf{1})$ & $145.825$ & $138.009$ \\\hline
    \end{tabular}
    \caption{Simulation time of the P system with/without garbage rules.}
    \label{table:simulation}
\end{table}

\section{Conclusion}
\label{section:conclusion}

In this article, we improved the system solving $\SAT$ from \cite{VCAl15} to get a cell-like symport/antiport P system with membrane separation and rules of length at most 3 solving the $\MIDSAT$ problem.
Because this problem is $\PsharpP$-complete, we obtained the inclusion $\PsharpP \subseteq \PMCCSCk{3}$.

We noticed that these systems are indeed a specific case of tissue P system, already proved to characterize $\PsharpP$ in \cite{LepoAl17}.
This leads us to characterize $\PsharpP$ by the set of all problems solved in polynomial time by a polynomially uniform family of cell-like symport/antiport P systems with membrane separation and rules of length at most 3.

This result implies that the computational powers of cell-like and tissue-like symport/ antiport P systems with membrane separation running in polynomial time are the same.
Moreover, it is known that when cell-like symport/antiport systems use membrane division instead of membrane separation, they can solve problems beyond $\PSPACE$ \cite{Song15}.
Then, our characterization reveals the computational power of membrane division against membrane separation (as long as $\PsharpP \subset \PSPACE$).

\bibliographystyle{plain}
\bibliography{bibliography}

\begin{thebibliography}{10}

\bibitem{Alhazov2004}
Artiom Alhazov and Tseren-Onolt Ishdorj.
\newblock Membrane operations in {P} systems with active membranes.
\newblock In Gheorghe P{\u{a}}un, Agust{\'i}n Riscos-N{\'u}{\~{n}}ez, Alvaro Romero-Jimenez, and Fernando Sancho-Caparrini, editors, {\em Proceedings of the Second Brainstorming Week on Membrane Computing}, pages 37--44. University of Sevilla, 2004.

\bibitem{LepoAl17}
Alberto Leporati, Luca Manzoni, Giancarlo Mauri, Antonio~E. Porreca, and Claudio Zandron.
\newblock Characterising the complexity of tissue {P} systems with fission rules.
\newblock {\em Journal of Computer and System Sciences}, 90:115--128, 2017.

\bibitem{MRAl17}
Luis Mac{\'\i}as~Ramos, Bosheng Song, Tao Song, Linqiang Pan, and Mario de~Jes{\'u}s P{\'e}rez~Jim{\'e}nez.
\newblock Limits on efficient computation in {P} systems with symport/antiport rules.
\newblock {\em BWMC 2017: 15th Brainstorming Week on Membrane Computing (2017), p 147-160}, 2017.

\bibitem{MRAl15}
Luis Mac{\'\i}as-Ramos, Bosheng Song, Luis Valencia-Cabrera, Linqiang Pan, and Mario Pérez-Jiménez.
\newblock Membrane fission: A computational complexity perspective.
\newblock {\em Complexity}, 21, 04 2015.

\bibitem{PLinguaSA}
Luis Mac{\'\i}as-Ramos, Luis Valencia-Cabrera, Bosheng Song, Tao Song, Linqiang Pan, and Mario Pérez-Jiménez.
\newblock A {P}-{L}ingua based simulator for {P} systems with symport/antiport rules.
\newblock {\em Fundamenta Informaticae}, 139:211--227, 07 2015.

\bibitem{OMAl20}
David Orellana-Martín, Miguel~\`A. {Martínez-del-Amor}, Luis Valencia-Cabrera, Bosheng Song, Linqiang Pan, and Mario~J. Pérez-Jiménez.
\newblock P systems with symport/antiport rules: When do the surroundings matter?
\newblock {\em Theoretical Computer Science}, 805:206--217, 2020.

\bibitem{PP02}
Andrei P{\u{a}}un and Gheorghe P{\u{a}}un.
\newblock The power of communication: P systems with symport/antiport.
\newblock {\em New Generation Comput.}, 20:295--306, 09 2002.

\bibitem{Paun00}
Gheorghe P{\u{a}}un.
\newblock Computing with membranes.
\newblock {\em Journal of Computer and System Sciences}, 61(1):108--143, 2000.

\bibitem{PaunNP01}
Gheorghe P{\u{a}}un.
\newblock P systems with active membranes: Attacking $\mathbf{NP}$ complete problems.
\newblock {\em Journal of Automata, Languages and Combinatorics}, 6:75--90, 01 2001.

\bibitem{Paun06}
Gheorghe P{\u{a}}un.
\newblock {\em Introduction to Membrane Computing}, pages 1--42.
\newblock Springer Berlin Heidelberg, Berlin, Heidelberg, 2006.

\bibitem{HBMC10}
Gheorghe P{\u{a}}un, Grzegorz Rozenberg, and Arto Salomaa.
\newblock {\em The Oxford Handbook of Membrane Computing}.
\newblock Oxford University Press, Inc., USA, 2010.

\bibitem{PLingua}
Ignacio P{\'e}rez-Hurtado, Luis Valencia-Cabrera, Mario~J. P{\'e}rez-Jim{\'e}nez, M.~Angels Colomer, and Agust{\'\i}n Riscos-N{\'u}{\~n}ez.
\newblock Mecosim: A general purpose software tool for simulating biological phenomena by means of {P} systems.
\newblock {\em IEEE Fifth International Conference on Bio-inpired Computing: Theories and Applications (BIC-TA 2010)}, I:637--643, 2010.

\bibitem{PJ10}
Mario~J. P{\'e}rez-Jim{\'e}nez.
\newblock A computational complexity theory in membrane computing.
\newblock In Gheorghe P{\u{a}}un, Mario~J. P{\'e}rez-Jim{\'e}nez, Agust{\'i}n Riscos-N{\'u}{\~{n}}ez, Grzegorz Rozenberg, and Arto Salomaa, editors, {\em Membrane Computing}, pages 125--148, Berlin, Heidelberg, 2010. Springer Berlin Heidelberg.

\bibitem{Song15}
Bosheng Song, Mario~J. Pérez-Jiménez, and Linqiang Pan.
\newblock Efficient solutions to hard computational problems by {P} systems with symport/antiport rules and membrane division.
\newblock {\em Biosystems}, 130:51--58, 2015.

\bibitem{Sosik19}
Petr Sosík.
\newblock P systems attacking hard problems beyond $\mathbf{NP}$: a survey.
\newblock {\em Journal of Membrane Computing}, 1, 09 2019.

\bibitem{Toda94}
Seinosuke Toda.
\newblock Simple characterizations of $\mathbf{P(\#P)}$ and complete problems.
\newblock {\em Journal of Computer and System Sciences}, 49(1):1--17, 1994.

\bibitem{VCAl15}
Luis Valencia-Cabrera, Bosheng Song, Luis Mac{\'\i}as-Ramos, Linqiang Pan, Agust{\'\i}n Riscos-N{\'u}{\~n}ez, and Mario~J. P{\'e}rez-Jim{\'e}nez.
\newblock Computational efficiency of {P} systems with symport/antiport rules and membrane separation.
\newblock {\em Thirteenth Brainstorming Week on Membrane Computing}, 1:325--370, 02 2015.

\end{thebibliography}

\newpage

\appendix

\section{Complete Definition of the P System}
\label{appendix:definition}

The P system family $\mathbf{\Pi} = (\Pi_{\langle m, n\rangle})_{m, n}$ we designed is defined by:
$${\Pi}_{\langle m, n \rangle} = (\Gamma', (\Gamma' \setminus \Gamma_1, \Gamma_1), \Sigma, \mathcal{E}', [ [\ ]_2 ]_1, \mathcal{M}'_1, \mathcal{M}_2, \mathcal{R}'_1, \mathcal{R}'_2, 1, 0)$$

where the system solving $\SAT$ created by \cite{VCAl15} is:
$$\Pi^{\mathbf{SAT}}_{\langle m, n \rangle} = (\Gamma, (\Gamma \setminus \Gamma_1, \Gamma_1), \Sigma, \mathcal{E}, [ [\ ]_2 [\ ]_3 ]_1, \mathcal{M}_1, \mathcal{M}_2, \mathcal{M}_3, \mathcal{R}_1, \mathcal{R}_2, \mathcal{R}_3, 1, 0)$$

and where the modified parameters are:
\begin{itemize}
    \item $\mathcal{R}'_1 = (\mathcal{R}_1 \setminus D_1) \cup \mathcal{R}''_1$,
    \item $\mathcal{R}'_2 = (\mathcal{R}_2 \setminus D_2) \cup \mathcal{R}''_2$,
    \item $\mathcal{E}' = (\mathcal{E} \setminus D) \cup \mathcal{E}''$,
    \item $\Gamma' = (\Gamma \setminus D) \cup \mathcal{E}'' \cup \{ \xi_{i,0} \}_{i \in \intinter{1}{n}} \cup \{ \omega_0 \}$,
    \item $\mathcal{M}'_1 = \mathcal{M}_1 - \{ f_0 \} - \{ f'_p \}_{p \in \intinter{1}{\DeltaSAT+1}} + \{ \mathbf{no}, \omega_0 \} + \{ \xi_{i,0} \}_{i \in \intinter{1}{n}}$.
\end{itemize}

Membrane 3 has been removed, together with its rules $\mathcal{R}_3$ and initial multiset $\mathcal{M}_3$.

We also removed the following objects and rules:
\begin{itemize}
    \item $D = \{ f_r \}_{r \in \intinter{0}{\DeltaSAT}} \cup \{ f'_p \}_{p \in \intinter{0}{\DeltaSAT+1}} \subseteq \Gamma$,
    \item $D_1 = \{\sympout{E_0 f_{\DeltaSAT} \mathbf{yes}}, \sympout{f_{\DeltaSAT} \mathbf{no}}\} \subseteq \mathcal{R}_1$,
    \item $D_2 = \{\sympout{e_{i,m}E_0}, \sympout{\bar{e}_{i,m}E_0}\}_{i \in \intinter{1}{n}} \subseteq \mathcal{R}_2$.
\end{itemize}

We added the objects of $\mathcal{E}''$ in the environment and the rules of $\mathcal{R}''_1$ and $\mathcal{R}''_2$:

\begin{fleqn}
\begin{align*}
    \mathcal{E}'' = & \ \{ \xi_{i,s} \}_{i \in \intinter{1}{n}, s \in \intinter{1}{\IOUT-1}} \cup \{ \xi_{1,\IOUT} \} \cup \{ \omega_s \}_{s \in \intinter{1}{\IOUT - \DeltaDete}} \\
    & \cup \{ \omega'_s \}_{s \in \intinter{\DeltaSAT-1}{\DeltaSAT + n-2}} \cup \{ \omega'_{\IComp(k) + 2i} \}_{i \in \intinter{0}{n-2}, k \in \intinter{1}{n}} \\
    & \cup \{ \theta_{\IComp(k) + j-1} \}_{j \in \intinter{0}{2n-1}, k \in \intinter{1}{n}} \cup \{ \theta'_{\IComp(k) + 2n -4} \}_{k \in \intinter{1}{n}} \\
    & \cup \{ \sigma \} \cup \{ \sigma'_i \}_{i \in \intinter{0}{n-1}} \cup \{ \Lambda_i, \Psi_i \}_{i \in \intinter{1}{n}} \cup \{ \nu_{i,j} \}_{i \in \intinter{1}{n}, j \in \intinter{0}{n-1}} \\
    & \cup \{ l^\Lambda_{k,i}, g^\Lambda_{k,i}, p^\Lambda_{k,i}, l^\Psi_{k,i}, g^\Psi_{k,i}, p^\Psi_{k,i}, l_{k,i}, g_{k,i}, p_{k,i} \}_{k \in \intinter{1}{n}, i \in \intinter{1}{k}} \\
    & \cup \{ \bar{l}^\Lambda_{k,i}, \bar{g}^\Lambda_{k,i}, \bar{p}^\Lambda_{k,i}, \bar{l}^\Psi_{k,i}, \bar{g}^\Psi_{k,i}, \bar{p}^\Psi_{k,i}, \bar{l}_{k,i}, \bar{g}_{k,i}, \bar{p}_{k,i} \}_{k \in \intinter{1}{n}, i \in \intinter{1}{n}, i \neq k} \\
    & \cup \{ \lambda_{k,i}, \pi_{k,i}, \zeta_{k,i} \}_{k \in \intinter{1}{n}, i \in \intinter{1}{k}} \cup \{ \bar{\lambda}_{k,i}, \bar{\pi}_{k,i}, \bar{\zeta}_{k,i} \}_{k \in \intinter{1}{n}, i \in \intinter{1}{n}, i \neq k} \\
    & \cup \{ l'_{k,i}, g'_{k,i}, p'_{k,i} \}_{k \in \intinter{1}{n}, i \in \intinter{1}{n-1} } \cup \{ \zeta'_{k,i} \}_{k \in \intinter{1}{n}, i \in \intinter{1}{n}} \\
    & \cup \{ \eta_{k,i,j} \}_{k \in \intinter{1}{n}, i \in \intinter{1}{k}, j \in \intinter{0}{n}} \cup \{ \bar{\eta}_{k,i,j} \}_{k \in \intinter{1}{n}, i \in \intinter{1}{n}, i \neq k, j \in \intinter{0}{n}} \\
    & \cup \{ \pi'_k, \phi_k, \phi'_k, \phi''_k \}_{k \in \intinter{1}{n}} \cup \{ \kappa_k, \kappa'_k, \kappa''_k, \kappa'''_k \}_{k \in \intinter{1}{n}} \\
    & \cup \{ \chi_k, \mu_k, \bar{\mu}_k, L_k, G_k \}_{k \in \intinter{1}{n}}
\end{align*}
\end{fleqn}

\begin{fleqn}
\begin{align*}
    & \begin{rcases}
        \mathcal{R}''_1 = & \ \{ \antiport{\xi_{i,s}}{\xi_{i,s+1}} \}_{i \in \intinter{1}{n},s \in \intinter{0}{\DeltaSAT - 1}} \\
        \phantom{\mathcal{R}''_1 =}
        & \cup \{ \antiport{\xi_{i,\DeltaSAT}}{\xi_{i,\DeltaSAT + 1},\nu_{i,0}} \}_{i \in \intinter{1}{n}} \\
        & \cup \{ \antiport{\xi_{i,s}}{\xi_{i,s+1}} \}_{i \in \intinter{1}{n},s \in \intinter{\DeltaSAT + 1}{\IPrep(1) - 2}} \\
        & \cup \{ \antiport{\xi_{i,\IPrep(k) - 1}}{\xi_{i,\IPrep(k)},\zeta'_{k,i}} \}_{i \in \intinter{1}{n},k \in \intinter{1}{n}} \\
        & \cup \{ \antiport{\xi_{i,\IPrep(k) + j}}{\xi_{i,\IPrep(k) + j + 1}} \}_{i \in \intinter{1}{n},k \in \intinter{1}{n}}^{j \in \intinter{0}{\DeltaLoop - 3}} \\
        & \cup \{ \antiport{\xi_{1,\IPrep(k+1) - 2}}{\xi_{1,\IPrep(k+1) - 1},\bar{\mu}_{k}} \}_{k \in \intinter{1}{n}} \\
        & \cup \{ \antiport{\xi_{i,\IPrep(k+1) - 2}}{\xi_{i,\IPrep(k+1) - 1}} \}_{i \in \intinter{2}{n},k \in \intinter{1}{n}} \\
        & \cup \{ \antiport{\xi_{1,\IOUT - 1}}{\xi_{1,\IOUT}} \}
    \end{rcases}
    {\begin{array}{l}
        \text{Rules for} \\
        \text{counters $\xi_i$.}
    \end{array}} \\
    & \begin{rcases}
        \phantom{\mathcal{R}''_1 =}
        & \cup \{ \antiport{\omega_{s}}{\omega_{s+1}^2} \}_{s \in \intinter{0}{n-1}} \\
        & \cup \{ \antiport{\omega_{s}}{\omega_{s+1}} \}_{s \in \intinter{n}{\DeltaSAT - 3}} \\
        & \cup \{ \antiport{\omega_{\DeltaSAT - 2}}{\omega_{\DeltaSAT - 1}, \omega'_{\DeltaSAT - 1}} \} \\
        & \cup \{ \antiport{\omega_{\DeltaSAT - 1}}{\omega_{\DeltaSAT}} \} \\
        & \cup \{ \antiport{\omega_{s}}{\omega_{s+1}, \sigma} \}_{s \in \intinter{\DeltaSAT}{\DeltaSAT + n - 1}} \\
        & \cup \{ \antiport{\omega_{s}}{\omega_{s+1}} \}_{s \in \intinter{\DeltaSAT + n}{\IComp(1) - 3}}
    \end{rcases}
    {\begin{array}{l}
        \text{Rules for counter $\omega$} \\
        \text{until first comparison} \\
        \text{sub-phase.}
    \end{array}} \\
    & \begin{rcases}
        \phantom{\mathcal{R}''_1 =}
        & \cup \{ \antiport{\omega_{\IComp(k) - 2}}{\omega_{\IComp(k) - 1}, \theta_{\IComp(k) - 1}} \}_{k \in \intinter{1}{n}} \\
        & \cup \{ \antiportbr{\omega_{\IComp(k) + 2i - 1}}{\omega_{\IComp(k) + 2i}, \omega'_{\IComp(k) + 2i}} \}_{i \in \intinter{0}{n-2}, k \in \intinter{1}{n}} \\
        & \cup \{ \antiportbr{\omega_{\IComp(k) + 2i}}{\omega_{\IComp(k) + 2i + 1}, p'_{k, i+1}} \}_{i \in \intinter{0}{n-2}, k \in \intinter{1}{n}} \\
        & \cup \{ \antiport{\omega_{\IComp(k) + 2n - 3}}{\omega_{\IComp(k) + 2n - 2}} \}_{k \in \intinter{1}{n}} \\
        & \cup \{ \antiport{\omega_{\IComp(k) + 2n - 2}}{\omega_{\IComp(k) + 2n - 1}} \}_{k \in \intinter{1}{n}} \\
        & \cup \{ \antiport{\omega_{\IComp(k) + 2n - 1}}{\omega_{\IComp(k) + 2n}, \kappa'''_{k}} \}_{k \in \intinter{1}{n}}
    \end{rcases}
    {\begin{array}{l}
        \text{Rules for counter} \\
        \text{$\omega$, $k$-th comparison} \\
        \text{sub-phase.}
    \end{array}} \\
    & \begin{rcases}
        \phantom{\mathcal{R}''_1 =}
        & \cup \{ \antiport{\omega_{\IDete(k) - 1}}{\omega_{\IDete(k)}, \chi_{k}} \}_{k \in \intinter{1}{n}} \\
        & \cup \{ \antiportbr{\omega_{\IDete(k) + j}}{\omega_{\IDete(k) + j+1}} \}_{k \in \intinter{1}{n-1}}^{j \in \intinter{0}{\DeltaDete + \DeltaPrep - 3}}
    \end{rcases}
    {\begin{array}{l}
        \text{Rules for counter $\omega$,} \\
        \text{$k$-th determination and} \\
        \text{$k$+1-th preparation} \\
        \text{sub-phases.}
    \end{array}} \\
    & \begin{rcases}
        \phantom{\mathcal{R}''_1 =}
        & \cup \{ \antiport{\omega'_{\DeltaSAT + s-1}}{\omega'_{\DeltaSAT + s}, \sigma'_{s}} \}_{s \in \intinter{0}{n-2}} \\
        & \cup \{ \antiport{\omega'_{\DeltaSAT + n-2}}{\sigma'_{n-1}} \} \\
        & \cup \{ \antiport{\nu_{i, j}}{\nu_{i, j+1}^2} \}_{i \in \intinter{1}{n}, j \in \intinter{0}{n-2}} \\
        & \cup \{ \antiport{\nu_{i, n-1}}{\Lambda_{i}, \Psi_{i}} \}_{i \in \intinter{1}{n}}
    \end{rcases}
    {\begin{array}{l}
        \text{Rules for the} \\
        \text{activation phase.}
    \end{array}} \\
    & \begin{rcases}
        \phantom{\mathcal{R}''_1 =}
        & \cup \{ \antiport{\zeta'_{k, i}, \mu_{i}}{\zeta_{k, i}} \}_{i \in \intinter{1}{k-1}, k \in \intinter{1}{n}} \\
        & \cup \{ \antiport{\zeta'_{k, i}, \bar{\mu}_{i}}{\bar{\zeta}_{k, i}} \}_{i \in \intinter{1}{k-1}, k \in \intinter{1}{n}} \\
        & \cup \{ \antiport{\zeta'_{k, k}}{\zeta_{k, k}} \}_{k \in \intinter{1}{n}} \\
        & \cup \{ \antiport{\zeta'_{k, i}}{\bar{\zeta}_{k, i}} \}_{i \in \intinter{k+1}{n}, k \in \intinter{1}{n}} \\
        & \cup \{ \antiport{\zeta_{k, i}}{\eta_{k, i, 0}, \mu_{i}} \}_{i \in \intinter{1}{k-1}, k \in \intinter{1}{n}} \\
        & \cup \{ \antiport{\bar{\zeta}_{k, i}}{\bar{\eta}_{k, i, 0}, \bar{\mu}_{i}} \}_{i \in \intinter{1}{k-1}, k \in \intinter{1}{n}} \\
        & \cup \{ \antiport{\zeta_{k, k}}{\eta_{k, k, 0}} \}_{k \in \intinter{1}{n}} \\
        & \cup \{ \antiport{\bar{\zeta}_{k, i}}{\bar{\eta}_{k, i, 0}} \}_{i \in \intinter{k+1}{n}, k \in \intinter{1}{n}} \\
        & \cup \{ \antiport{\eta_{k, i, j}}{\eta_{k, i, j+1}^2} \}_{i \in \intinter{1}{k}, k \in \intinter{1}{n}, j \in \intinter{0}{n-1}} \\
        & \cup \{ \antiport{\bar{\eta}_{k, i, j}}{\bar{\eta}_{k, i, j+1}^2} \}_{k \in \intinter{1}{n}, i \in \intinter{1}{n}, j \in \intinter{0}{n-1}, i \neq k} \\
        & \cup \{ \antiport{\eta_{k, i, n}}{\lambda_{k, i}, \pi_{k, i}} \}_{i \in \intinter{1}{k}, k \in \intinter{1}{n}} \\
        & \cup \{ \antiport{\bar{\eta}_{k, i, n}}{\bar{\lambda}_{k, i}, \bar{\pi}_{k, i}} \}_{k \in \intinter{1}{n}, i \in \intinter{1}{n}, i \neq k}
    \end{rcases}
    {\begin{array}{l}
        \text{Rules for the} \\
        \text{preparation} \\
        \text{sub-phase,} \\
        \text{part I.}
    \end{array}} \\
    & \begin{rcases}
        \phantom{\mathcal{R}''_1 =}
        & \cup \{ \antiport{\lambda_{k, i}}{l_{k, i}, g_{k, i}} \}_{i \in \intinter{1}{k}, k \in \intinter{1}{n}} \\
        & \cup \{ \antiport{\bar{\lambda}_{k, i}}{\bar{l}_{k, i}, \bar{g}_{k, i}} \}_{k \in \intinter{1}{n}, i \in \intinter{1}{n}, i \neq k} \\
        & \cup \{ \antiport{\pi_{k, i}}{p_{k, i}, \pi'_{k}} \}_{i \in \intinter{1}{k}, k \in \intinter{1}{n}} \\
        & \cup \{ \antiport{\bar{\pi}_{k, i}}{\bar{p}_{k, i}, \pi'_{k}} \}_{k \in \intinter{1}{n}, i \in \intinter{1}{n}, i \neq k} \\
        & \cup \{ \antiport{l_{k, i}}{l^{\Lambda}_{k, i}, l^{\Psi}_{k, i}} \}_{i \in \intinter{1}{k}, k \in \intinter{1}{n}} \\
        & \cup \{ \antiport{g_{k, i}}{g^{\Lambda}_{k, i}, g^{\Psi}_{k, i}} \}_{i \in \intinter{1}{k}, k \in \intinter{1}{n}} \\
        & \cup \{ \antiport{p_{k, i}}{p^{\Lambda}_{k, i}, p^{\Psi}_{k, i}} \}_{i \in \intinter{1}{k}, k \in \intinter{1}{n}} \\
        & \cup \{ \antiport{\bar{l}_{k, i}}{\bar{l}^{\Lambda}_{k, i}, \bar{l}^{\Psi}_{k, i}} \}_{k \in \intinter{1}{n}, i \in \intinter{1}{n}, i \neq k} \\
        & \cup \{ \antiport{\bar{g}_{k, i}}{\bar{g}^{\Lambda}_{k, i}, \bar{g}^{\Psi}_{k, i}} \}_{k \in \intinter{1}{n}, i \in \intinter{1}{n}, i \neq k} \\
        & \cup \{ \antiport{\bar{p}_{k, i}}{\bar{p}^{\Lambda}_{k, i}, \bar{p}^{\Psi}_{k, i}} \}_{k \in \intinter{1}{n}, i \in \intinter{1}{n}, i \neq k} \\
        & \cup \{ \antiport{\pi'_{k}}{\phi_{k}, \phi'_{k}} \}_{k \in \intinter{1}{n}}
    \end{rcases}
    {\begin{array}{l}
        \text{Rules for the} \\
        \text{preparation} \\
        \text{sub-phase,} \\
        \text{part II.}
    \end{array}} \\
    & \begin{rcases}
        \phantom{\mathcal{R}''_1 =}
        & \cup \{ \antiport{\omega'_{\IComp(k) + 2i}}{l'_{k, i+1}, g'_{k, i+1}} \}_{i \in \intinter{0}{n-2}, k \in \intinter{1}{n}} \\
        & \cup \{ \antiport{\phi'_{k}}{\phi''_{k}} \}_{k \in \intinter{1}{n}} \\
        & \cup \{ \sympout{\phi_{k}, \phi''_{k}} \}_{k \in \intinter{1}{n}} \\
        & \cup \{ \antiport{\theta_{\IComp(k) + j-1}}{\theta_{\IComp(k) + j}} \}_{k \in \intinter{1}{n}, j \in \intinter{0}{2n - 5}} \\
        & \cup \{ \antiportbr{\theta_{\IComp(k) + 2n - 5}}{\theta_{\IComp(k) + 2n - 4}, \theta'_{\IComp(k) + 2n - 4}} \}_{k \in \intinter{1}{n}} \\
        & \cup \{ \antiport{\theta_{\IComp(k) + 2n - 4}}{\theta_{\IComp(k) + 2n - 3}^2} \}_{k \in \intinter{1}{n}} \\
        & \cup \{ \antiport{\theta'_{\IComp(k) + 2n - 4}}{\theta_{\IComp(k) + 2n - 3}} \}_{k \in \intinter{1}{n}} \\
        & \cup \{ \antiport{\theta_{\IComp(k) + 2n - 3}}{\theta_{\IComp(k) + 2n - 2}^2} \}_{k \in \intinter{1}{n}} \\
        & \cup \{ \antiport{\theta_{\IComp(k) + 2n - 2}}{\kappa_{k}, \kappa'_{k}} \}_{k \in \intinter{1}{n}}
    \end{rcases}
    {\begin{array}{l}
        \text{Rules for the} \\
        \text{comparison} \\
        \text{sub-phase,} \\
        \text{part I.}
    \end{array}} \\
    & \begin{rcases}
        \phantom{\mathcal{R}''_1 =}
        & \cup \{ \antiport{\kappa'_{k}}{\kappa''_{k}} \}_{k \in \intinter{1}{n}} \\
        & \cup \{ \sympout{\kappa_{n}, l^{\Lambda}_{n, n}}, \sympout{\kappa_{n}, l^{\Psi}_{n, n}} \} \\
        & \cup \{ \sympout{\kappa_{n}, g^{\Lambda}_{n, n}}, \sympout{\kappa_{n}, g^{\Psi}_{n, n}} \} \\
        & \cup \{ \sympout{\kappa_{n}, p^{\Lambda}_{n, n}}, \sympout{\kappa_{n}, p^{\Psi}_{n, n}} \} \\
        & \cup \{ \sympout{\kappa_{k}, \bar{l}^{\Lambda}_{k, n}}, \sympout{\kappa_{k}, \bar{l}^{\Psi}_{k, n}} \}_{k \in \intinter{1}{n-1}} \\
        & \cup \{ \sympout{\kappa_{k}, \bar{g}^{\Lambda}_{k, n}}, \sympout{\kappa_{k}, \bar{g}^{\Psi}_{k, n}} \}_{k \in \intinter{1}{n-1}} \\
        & \cup \{ \sympout{\kappa_{k}, \bar{p}^{\Lambda}_{k, n}}, \sympout{\kappa_{k}, \bar{p}^{\Psi}_{k, n}} \}_{k \in \intinter{1}{n-1}} \\
        & \cup \{ \sympout{\kappa_{k}, \kappa''_{k}} \}_{k \in \intinter{1}{n}}
    \end{rcases}
    {\begin{array}{l}
        \text{Rules for the} \\
        \text{comparison} \\
        \text{sub-phase,} \\
        \text{part II.}
    \end{array}} \\
    & \begin{rcases}
        \phantom{\mathcal{R}''_1 =}
        & \cup \{ \antiport{l^{\Lambda}_{n, n}, \chi_{n}}{L_{n}}, \antiport{l^{\Psi}_{n, n}, \chi_{n}}{L_{n}} \} \\
        & \cup \{ \antiport{g^{\Lambda}_{n, n}, \chi_{n}}{G_{n}}, \antiport{g^{\Psi}_{n, n}, \chi_{n}}{G_{n}} \} \\
        & \cup \{ \antiport{p^{\Lambda}_{n, n}, \chi_{n}}{G_{n}}, \antiport{p^{\Psi}_{n, n}, \chi_{n}}{G_{n}} \} \\
        & \cup \{ \antiport{\bar{l}^{\Lambda}_{k, n}, \chi_{k}}{L_{k}}, \antiport{\bar{l}^{\Psi}_{k, n}, \chi_{k}}{L_{k}} \}_{k \in \intinter{1}{n-1}} \\
        & \cup \{ \antiport{\bar{g}^{\Lambda}_{k, n}, \chi_{k}}{G_{k}}, \antiport{\bar{g}^{\Psi}_{k, n}, \chi_{k}}{G_{k}} \}_{k \in \intinter{1}{n-1}} \\
        & \cup \{ \antiport{\bar{p}^{\Lambda}_{k, n}, \chi_{k}}{G_{k}}, \antiport{\bar{p}^{\Psi}_{k, n}, \chi_{k}}{G_{k}} \}_{k \in \intinter{1}{n-1}} \\
        & \cup \{ \sympout{G_{k}, L_{k}} \}_{k \in \intinter{1}{n}} \\
        & \cup \{ \antiport{G_{k}, \bar{\mu}_{k}}{\mu_{k}} \}_{k \in \intinter{1}{n}}
    \end{rcases}
    {\begin{array}{l}
        \text{Rules for the} \\
        \text{determination} \\
        \text{sub-phase.}
    \end{array}} \\
    & \begin{rcases}
        \phantom{\mathcal{R}''_1 =}
        & \cup \{ \sympout{\mu_{n}, \mathbf{yes}, \xi_{1, \IOUT}} \} \\
        & \cup \{ \sympout{\bar{\mu}_{n}, \mathbf{no}, \xi_{1, \IOUT}} \}
    \end{rcases}
    {\begin{array}{l}
        \text{Rules for the output phase.}
    \end{array}}  \\
\end{align*}
\end{fleqn}

\begin{fleqn}
\begin{align*}
    & \begin{rcases}
        \mathcal{R}''_2 = & \ \{ \antiport{e_{i, m}}{\sigma'_{0}, T_{i}} \}_{i \in \intinter{1}{n}} \\
        \phantom{\mathcal{R}''_2 =}
        & \cup \{ \antiport{\bar{e}_{i, m}}{\sigma'_{0}, F_{i}} \}_{i \in \intinter{1}{n}} \\
        & \cup \{ \antiport{\sigma'_{i}}{\sigma'_{i+1}, \sigma} \}_{i \in \intinter{0}{n-2}} \\
        & \cup \{ \antiport{\sigma'_{n-1}}{\sigma} \} \\
        & \cup \{ \antiport{\sigma, T_{i}}{\Lambda_{i}} \}_{i \in \intinter{1}{n}} \\
        & \cup \{ \antiport{\sigma, F_{i}}{\Psi_{i}} \}_{i \in \intinter{1}{n}}
    \end{rcases}\text{Rules for the activation phase.} \\
    & \begin{rcases}
        \phantom{\mathcal{R}''_2 =}
        & \cup \{ \antiport{\Lambda_{1}}{p^{\Lambda}_{k, 1}, \phi_{k}} \}_{k \in \intinter{1}{n}} \\
        & \cup \{ \antiport{\Psi_{1}}{l^{\Psi}_{k, 1}, \phi_{k}} \}_{k \in \intinter{1}{n}} \\
        & \cup \{ \antiport{\Lambda_{1}}{\bar{g}^{\Lambda}_{k, 1}, \phi_{k}} \}_{k \in \intinter{2}{n}} \\
        & \cup \{ \antiport{\Psi_{1}}{\bar{p}^{\Psi}_{k, 1}, \phi_{k}} \}_{k \in \intinter{2}{n}}
    \end{rcases}
    {\begin{array}{l}
        \text{Rules for the comparison sub-phase,} \\
        \text{applied at step $\IComp(k) + 1$.}
    \end{array}} \\
    & \begin{rcases}
        \phantom{\mathcal{R}''_2 =}
        & \cup \{ \antiport{l'_{k, i-1}, \Lambda_{i}}{l^{\Lambda}_{k, i}} \}_{i \in \intinter{2}{k}, k \in \intinter{1}{n}} \\
        & \cup \{ \antiport{l'_{k, i-1}, \Psi_{i}}{l^{\Psi}_{k, i}} \}_{i \in \intinter{2}{k}, k \in \intinter{1}{n}} \\
        & \cup \{ \antiport{l'_{k, i-1}, \Lambda_{i}}{\bar{l}^{\Lambda}_{k, i}} \}_{k \in \intinter{1}{n}, i \in \intinter{2}{n}, i \neq k} \\
        & \cup \{ \antiport{l'_{k, i-1}, \Psi_{i}}{\bar{l}^{\Psi}_{k, i}} \}_{k \in \intinter{1}{n}, i \in \intinter{2}{n}, i \neq k} \\
        & \cup \{ \antiport{g'_{k, i-1}, \Lambda_{i}}{g^{\Lambda}_{k, i}} \}_{i \in \intinter{2}{k}, k \in \intinter{1}{n}} \\
        & \cup \{ \antiport{g'_{k, i-1}, \Psi_{i}}{g^{\Psi}_{k, i}} \}_{i \in \intinter{2}{k}, k \in \intinter{1}{n}} \\
        & \cup \{ \antiport{g'_{k, i-1}, \Lambda_{i}}{\bar{g}^{\Lambda}_{k, i}} \}_{k \in \intinter{1}{n}, i \in \intinter{2}{n}, i \neq k} \\
        & \cup \{ \antiport{g'_{k, i-1}, \Psi_{i}}{\bar{g}^{\Psi}_{k, i}} \}_{k \in \intinter{1}{n}, i \in \intinter{2}{n}, i \neq k} \\
        & \cup \{ \antiport{p'_{k, i-1}, \Lambda_{i}}{p^{\Lambda}_{k, i}} \}_{i \in \intinter{2}{k}, k \in \intinter{1}{n}} \\
        & \cup \{ \antiport{p'_{k, i-1}, \Psi_{i}}{l^{\Psi}_{k, i}} \}_{i \in \intinter{2}{k}, k \in \intinter{1}{n}} \\
        & \cup \{ \antiport{p'_{k, i-1}, \Lambda_{i}}{\bar{g}^{\Lambda}_{k, i}} \}_{k \in \intinter{1}{n}, i \in \intinter{2}{n}, i \neq k} \\
        & \cup \{ \antiport{p'_{k, i-1}, \Psi_{i}}{\bar{p}^{\Psi}_{k, i}} \}_{k \in \intinter{1}{n}, i \in \intinter{2}{n}, i \neq k}
    \end{rcases}{\begin{array}{l}
        \text{Rules for the} \\
        \text{comparison sub-phase,} \\
        \text{applied at steps} \\
        \text{$\IComp(k) + 2i - 1$.}
    \end{array}} \\
    & \begin{rcases}
        \phantom{\mathcal{R}''_2 =}
        & \cup \{ \antiport{l^{\Lambda}_{k, i}}{\Lambda_{i}, l'_{k, i}} \}_{i \in \intinter{1}{k}, k \in \intinter{1}{n-1}} \\
        & \cup \{ \antiport{l^{\Lambda}_{n, i}}{\Lambda_{i}, l'_{n, i}} \}_{i \in \intinter{1}{n-1}} \\
        & \cup \{ \antiport{l^{\Psi}_{k, i}}{\Psi_{i}, l'_{k, i}} \}_{i \in \intinter{1}{k}, k \in \intinter{1}{n-1}} \\
        & \cup \{ \antiport{l^{\Psi}_{n, i}}{\Psi_{i}, l'_{n, i}} \}_{i \in \intinter{1}{n-1}} \\
        & \cup \{ \antiport{g^{\Lambda}_{k, i}}{\Lambda_{i}, g'_{k, i}} \}_{i \in \intinter{1}{k}, k \in \intinter{1}{n-1}} \\
        & \cup \{ \antiport{g^{\Lambda}_{n, i}}{\Lambda_{i}, g'_{n, i}} \}_{i \in \intinter{1}{n-1}} \\
        & \cup \{ \antiport{g^{\Psi}_{k, i}}{\Psi_{i}, g'_{k, i}} \}_{i \in \intinter{1}{k}, k \in \intinter{1}{n-1}} \\
        & \cup \{ \antiport{g^{\Psi}_{n, i}}{\Psi_{i}, g'_{n, i}} \}_{i \in \intinter{1}{n-1}} \\
        & \cup \{ \antiport{p^{\Lambda}_{k, i}}{\Lambda_{i}, p'_{k, i}} \}_{i \in \intinter{1}{k}, k \in \intinter{1}{n-1}} \\
        & \cup \{ \antiport{p^{\Lambda}_{n, i}}{\Lambda_{i}, p'_{n, i}} \}_{i \in \intinter{1}{n-1}} \\
        & \cup \{ \antiport{p^{\Psi}_{k, i}}{\Psi_{i}, p'_{k, i}} \}_{i \in \intinter{1}{k}, k \in \intinter{1}{n-1}} \\
        & \cup \{ \antiport{p^{\Psi}_{n, i}}{\Psi_{i}, p'_{n, i}} \}_{i \in \intinter{1}{n-1}} \\
        & \cup \{ \antiport{\bar{l}^{\Lambda}_{k, i}}{\Lambda_{i}, l'_{k, i}} \}_{k \in \intinter{1}{n}, i \in \intinter{1}{n-1}, i \neq k} \\
        & \cup \{ \antiport{\bar{l}^{\Psi}_{k, i}}{\Psi_{i}, l'_{k, i}} \}_{k \in \intinter{1}{n}, i \in \intinter{1}{n-1}, i \neq k} \\
        & \cup \{ \antiport{\bar{g}^{\Lambda}_{k, i}}{\Lambda_{i}, g'_{k, i}} \}_{k \in \intinter{1}{n}, i \in \intinter{1}{n-1}, i \neq k} \\
        & \cup \{ \antiport{\bar{g}^{\Psi}_{k, i}}{\Psi_{i}, g'_{k, i}} \}_{k \in \intinter{1}{n}, i \in \intinter{1}{n-1}, i \neq k} \\
        & \cup \{ \antiport{\bar{p}^{\Lambda}_{k, i}}{\Lambda_{i}, p'_{k, i}} \}_{k \in \intinter{1}{n}, i \in \intinter{1}{n-1}, i \neq k} \\
        & \cup \{ \antiport{\bar{p}^{\Psi}_{k, i}}{\Psi_{i}, p'_{k, i}} \}_{k \in \intinter{1}{n}, i \in \intinter{1}{n-1}, i \neq k}
    \end{rcases}{\begin{array}{l}
        \text{Rules for the} \\
        \text{comparison sub-phase,} \\
        \text{applied at steps} \\
        \text{$\IComp(k) + 2i$.}
    \end{array}} \\
    & \begin{rcases}
        \phantom{\mathcal{R}''_2 =}
        & \cup \{ \antiport{l^{\Lambda}_{n, n}}{\Lambda_{n}, \kappa'''_{n}} \} \\
        & \cup \{ \antiport{l^{\Psi}_{n, n}}{\Psi_{n}, \kappa'''_{n}} \} \\
        & \cup \{ \antiport{g^{\Lambda}_{n, n}}{\Lambda_{n}, \kappa'''_{n}} \} \\
        & \cup \{ \antiport{g^{\Psi}_{n, n}}{\Psi_{n}, \kappa'''_{n}} \} \\
        & \cup \{ \antiport{p^{\Lambda}_{n, n}}{\Lambda_{n}, \kappa'''_{n}} \} \\
        & \cup \{ \antiport{p^{\Psi}_{n, n}}{\Psi_{n}, \kappa'''_{n}} \} \\
        & \cup \{ \antiport{\bar{l}^{\Lambda}_{k, n}}{\Lambda_{n}, \kappa'''_{k}} \}_{k \in \intinter{1}{n-1}} \\
        & \cup \{ \antiport{\bar{l}^{\Psi}_{k, n}}{\Psi_{n}, \kappa'''_{k}} \}_{k \in \intinter{1}{n-1}} \\
        & \cup \{ \antiport{\bar{g}^{\Lambda}_{k, n}}{\Lambda_{n}, \kappa'''_{k}} \}_{k \in \intinter{1}{n-1}} \\
        & \cup \{ \antiport{\bar{g}^{\Psi}_{k, n}}{\Psi_{n}, \kappa'''_{k}} \}_{k \in \intinter{1}{n-1}} \\
        & \cup \{ \antiport{\bar{p}^{\Lambda}_{k, n}}{\Lambda_{n}, \kappa'''_{k}} \}_{k \in \intinter{1}{n-1}} \\
        & \cup \{ \antiport{\bar{p}^{\Psi}_{k, n}}{\Psi_{n}, \kappa'''_{k}} \}_{k \in \intinter{1}{n-1}}
    \end{rcases}{\begin{array}{l}
        \text{Rules for the} \\
        \text{comparison sub-phase,} \\
        \text{applied at steps} \\
        \text{$\IComp(k) + 2n + 1$.}
    \end{array}}
\end{align*}
\end{fleqn}

\end{document}